\documentclass{article}

\usepackage{amsmath,amssymb,amsthm}
\usepackage[T1]{fontenc}
\usepackage{tikz}
\usepackage{lscape}
\usepackage{multirow}

\theoremstyle{definition}
\newtheorem{df}{Definition}[section]
\newtheorem{pr}{Example}[section]
\theoremstyle{plain}
\newtheorem{teo}{Theorem}

\title{Proximity-based equivalence classes in fuzzy relational database model}
\author{Aleksandar Janji\'c}
\begin{document}
\maketitle

\begin{abstract}

One of the first attempts to set a solid theoretical foundation for extending the content of relational databases with incomplete information was the fuzzy relational model by Buckles and Petry from 1982. This structure was based on two generalizations of the traditional relational model: (1) A tuple component is in general any subset of the corresponding domain, rather than a single element and (2) A \textit{similarity} relation is defined on each domain. This relation satisfies the properties of reflexivity, symmetry and max-min transitivity, thus having the equality relation as a special case. This generalization keeps two key properties of the relational model - that no two different tuples represent the same information and that the application of any operation of the relation algebra has a unique result.

At the end of the 1980s, Shenoi and Melton generalized this model and showed how its main characteristic - the existence of equivalence classes over the attribute domains - can also be preserved with a relation that only satisfies the properties of reflexivity and symmetry (\textit{proximity relation}). The motivation for this generalization is the strictness of the max-min transitivity property of similarity relations, which complicates the construction of this relation for some domain types.

An important characteristic of the Shenoi-Melton model is the dependence of the equivalence classes upon the so-called temporal domains, i.e. the current content of the database. This characteristic, together with the way the equivalence relation is constructed by the proximity relation, can lead to the equivalence classes that don't correspond well with some database query types, because they contain a large number of elements with very little mutual similarity.

Here we will present a different way of forming the equivalence classes over the attribute domains in fuzzy relational databases in which they depend only on the attribute domain and not on the current database state. We will also show a simple method for automatic construction of proximity relations over some domain types.

\end{abstract}

\section{Introduction}

The fuzzy relational database model of Buckles and Petry \cite{buckles1,petry1} is based on a weakened assumption of the definition of the first normal form in traditional relational databases. The request that the value of any attribute be a single element of the corresponding domain is relaxed by allowing it to be any non-empty subset of the domain. Moreover, with each domain is associated a \textit{similarity} relation that satisfies the properties of reflexivity, symmetry and max-min transitivity. This properties enable the forming of equivalence classes over each domain.

Shenoi and Melton \cite{shenoi1,shenoi2} showed a way of forming the equivalence classes without the similarity relation, by using a relation that satisfies only the reflexivity and symmetry properties (\textit{proximity} relation). The max-min transitivity property is very strict and on some domains it is hard to define in a natural way a relation that satisfies it.

An important characteristic of this model, namely that the equivalence classes depend on the so-called temporal domains, i.e. the current content of the database, can lead to equivalence classes that don't correspond well with some database query types, because they contain a large number of elements with very little mutual similarity. In this paper we will use proximity relations to define equivalence classes over fuzzy relational database domains that depend only on the domain and not on the current state of the database.

In the following chapter we will describe the fuzzy relational database model by Buckles and Petry and in the third chapter the Shenoi-Melton model. Afterwards, in the fourth chapter, we will show how equivalence classes can be formed by proximity relations on a linearly ordered one-dimensional domain. In Chapter 5 the method of Chapter 4 will be generalized to a two-dimensional domain. In Chapter six we will combine the examples used in the two previous chapters and show how equivalence classes are formed in relational database tables. Chapter seven contains a practical example demonstrating how different methods lead to different ways in which the equivalence classes are formed. Finally, in Chapter 8 we present our conclusions.

\section{Fuzzy relational database model}

\subsection{Relational database model}
First, we will very briefly describe the main features of the structural part of Codd's relational database model \cite{codd1} (the model consists of the structural, integrity and manipulative parts, but the letter two are of lesser importance for this paper). A \textit{relational scheme} is a set of \textit{attribute names} (or \textit{attributes} for short)$A_1,A_2,\dots,A_n$. To each attribute $A_i$, $1\leq i\leq n$ corresponds a set $D_i$, called the \textit{domain} of $A_i$. The domain of $A_i$ is usually denoted by $dom(A_i)$. The domains are non-empty, finite or countably infinite sets \cite{maier1}. Let $\mathbb{D}=D_1\cup D_2\cup\cdots D_n$. A \textit{tuple} $t$ on relation scheme $R=A_1,A_2,\dots,A_n$ is a mapping from $R$ to $\mathbb{D}$ such that $t(A_i)$ must be in $D_i$ for all $1\leq i\leq n$. A relation $r$ on $R$ is a collection of tuples on $R$. A relational \textit{database} is a collection of relations. 

Usually, relations are visually represented by way of rectangular tables. Each field is required to contain exactly one element, such that all elements in one column are of the same type. The columns of the table correspond to attributes, their data types to domains and the rows to tuples.

\subsection{Fuzzy relational database model by Buckles and Petry}

Buckles and Petry \cite{buckles1} have constructed one of the first fuzzy relational database model, with the goal to enable storing a wider class of unknown or imprecise data than was the case in the implementations of the traditional relational model. Those implementations were mostly limited to the possibility of assigning the NULL value \cite{codd1, zaniolo1} to a field in a table. NULL signifies the absence of the corresponding data, meaning that either the precise value is unknown or that the attribute itself is not applicable to the current tuple. Incomplete information like ''Tom's age is between 20 and 25'' could not be represented in those implementations.

The fuzzy relational database model of Buckles and Petry is based on two primary characteristics:

\begin{enumerate}
  \item The tuple components do not have to be singular values but can be arbitrary non-empty subsets of a previously given set.
  \item To each such set corresponds a similarity relation, which together with the \textit{merge} operation represents a generalization of the duplicate elimination based on equality relation from the traditional relational databases.
\end{enumerate}

These properties enable the operations of fuzzy relational algebra to be well defined.

\begin{df}
	
	Let non-empty sets $D_1$, $D_2$,$\dots$,$D_n$ be given and let $\mathbf{P}(D_i)=\mathcal{P}(D_i)\setminus\emptyset, i=\overline{1,n}$. A \textit{fuzzy tuple} over the sets $D_1$, $D_2$,$\dots$,$D_n$ is any element of the Cartesian product $\mathbf{P}(D_1)\times\mathbf{P}(D_2)\times\cdots\times\mathbf{P}(D_n)$.
	
\end{df}

Therefore, a fuzzy tuple $t$ is a set $t=(d_1,d_2,\dots,d_n)$, where $d_i\subset D_i$ for $i=\overline{1,n}$ and given sets $D_1,D_2$,$\dots,D_n$.

\begin{df} 
	Let $t=(d_1,d_2,\dots,d_n)$ be a tuple whose components $d_1,d_2,\dots,d_n$ are the subsets of the sets $D_1,D_2,\dots,D_n$, respectively. An \textit{interpretation} of a tuple $t$ is a tuple $\theta=(a_1,a_2,\dots,a_n)$ for which $a_i\in d_i$ for all
	$i=1,2,\dots,n$, i.e. each component of the tuple $\theta$ is a singleton subset of the corresponding tuple $t$.
\end{df}

\begin{df} 
	A \textit{fuzzy relation} over the sets $D_1$, $D_2$,$\dots$,$D_n$ is any set of fuzzy tuples over those sets, i.e. any subset of the Cartesian product $\mathbf{P}(D_1)\times\mathbf{P}(D_2)\times\cdots\times\mathbf{P}(D_n)$.
 	
\end{df}

As a generalization of the equality relation this model has a similarity relation over the sets $D_i$. This relation is based on a mapping which to each pair of elements of the set $D_i$ attaches a real value from the interval $[0,1]$ \cite{zadeh1}.

\begin{df}
	A similarity relation over a set $D$ is a mapping $s:D\times D\rightarrow[0,1]$ which satisfies the following properties:
	
\begin{enumerate}
  \item $s(x,x)=1$ for each element $x\in D$.
  \item $s(x,y)=s(y,x)$ for each pair of elements $x$ and $y$ from $D$.
  \item $s(x,z)\geq \min\{s(x,y),s(y,z)\}$ for any three elements $x$, $y$ and $z$ of the set $D$.
\end{enumerate}
\end{df}

In the relational database model each relation represents a set whose elements are the tuples of that relation. Defining the relation as a set means that any tuple appears in it at most once, i.e. that identical information is not present in different places at the same tame. Therefore, each operation of the relational algebra implies the elimination of duplicate tuples. The generalization of this concept in the fuzzy relational databases is based on the elimination of the \textit{redundant} tuples, depending on a previously given similarity threshold for each attribute's domain, $LEVEL(D)$, where $D$ is a domain. The elimination does not consist in deleting one of the redundant tuples, but their merging, i.e. the replacement of the existing tuples by a new tuple, each of whose components represents the union of the corresponding components of the given tuples.

\begin{df}
	
	Let $t_i=(d_{i1},d_{i2},\dots,d_{im})$, $i=1,2\dots,n$ be the tuples of a relation $R$, whose attributes have domains
	$D_1,D_2,\dots,D_m$, respectively. For the domain $D_j$, the \textit{similarity threshold} is the value
	
	\begin{equation*}
	THRES(D_j)=\min_{i=1,2,\dots,n}\{\min_{x,y\in d_{ij}} s_j(x,y)\},
	\end{equation*}
	where $s_j$ is the similarity relation of the domain $D_j$.
\end{df}

\begin{df} 
	Let $t_1=(d_{11},d_{12},\dots,d_{1m})$ and $t_2=(d_{21},d_{22},\dots,d_{2m})$ be tuples defined over the sets $D_1,D_2,\dots,D_m$. The tuples $t_1$ and $t_2$ are \textit{redundant} if for each domain $D_j$ $$LEVEL(D_j)\leq\min_{x,y\in d_{1j}\cup d_{2j}}s_j(x,y),$$ 
	where $LEVEL(D_j)$ is a given value in the interval $[0,1]$ and $s_j$ is the similarity relation over the set $D_j$.
\end{df}

It can be shown that the following theorems hold:

\begin{teo} 
	If a fuzzy relation $R$ contains no redundant tuples, then in $R$ there are no different tuples that have a common interpretation.
\end{teo}

\begin{teo} 
	Let $R$ be a fuzzy relation and $M$ the relation derived from $R$ by merging all the redundant tuples. Then $M$ is unique.
\end{teo}

\begin{teo}
	Redundancy is an equivalence relation.
\end{teo}

The first theorem guarantees that there is no ''overlapping'' of information in the tuples of a relation, namely that the existence of at least one common interpretation of two tuples implies complete redundancy of those tuples. Such tuples would be merged into one tuple whose components would contain all the elements of the corresponding components of the starting tuples. A special case of this process can be seen, for example, in the projection operation of the relational algebra. If after discarding a set of attributes in a relation some duplicate tuples remain, all but one of those tuples will be deleted from the relation. In the traditional relational databases all the attribute values in the tuples are single and the role of the similarity relation is played by strict equality. Thus, the elimination of redundant tuples is actually the merging of the tuples for the values $LEVEL(D_i)=1$ for each domain $D_i$, where $s$ is the equality relation.

Since the operations of the relational algebra in fuzzy relational databases are executed similarly to the traditional databases, except the final step, which contains the more general operation of tuple merging instead of the duplicate elimination, the second theorem guarantees that the result of each operation of the fuzzy relational algebra is unique, i.e. that the operations of the fuzzy relational algebra are well defined.

\begin{pr} \cite{petry1}
	Let $D_{Effect} =\{$Minimal, Limited, Tolerable, Moderate, Severe, Major, Extreme, Irreversible $\}$ be a set of linguistic values. These values represent the degrees of severity of the environmental polution caused by various poisonous substances or other conditions. Table \ref{simrelation} shows a similarity relation over this set.
	
	\begin{table}
		
		\begin{center}
			\begin{tabular}{|c|c|c|c|c|c|c|c|c|}
				\hline
				$s$ & Min & Lim & Tol & Mod & Sev & Maj & Ext & Irr \\
				\hline
				Minimal      & 1.00 & 0.90 & 0.90 & 0.85 & 0.75 & 0.75 & 0.75 & 0.00 \\
				Limited      & 0.90 & 1.00 & 0.95 & 0.85 & 0.75 & 0.75 & 0.75 & 0.00 \\
				Tolerable    & 0.90 & 0.95 & 1.00 & 0.85 & 0.75 & 0.75 & 0.75 & 0.00 \\
				Moderate     & 0.85 & 0.85 & 0.85 & 1.00 & 0.75 & 0.75 & 0.75 & 0.00 \\
				Severe       & 0.75 & 0.75 & 0.75 & 0.75 & 1.00 & 0.80 & 0.80 & 0.00 \\
				Major		 & 0.75 & 0.75 & 0.75 & 0.75 & 0.80 & 1.00 & 0.85 & 0.00 \\
				Extreme		 & 0.75 & 0.75 & 0.75 & 0.75 & 0.80 & 0.85 & 1.00 & 0.00 \\
				Irreversible & 0.00 & 0.00 & 0.00 & 0.00 & 0.00 & 0.00 & 0.00 & 1.00 \\
				\hline
			\end{tabular}
		\end{center}
		\caption{A similarity relation}
		\label{simrelation}
	\end{table}

Now, let the SURVEY relation be given, with the attributes Pollutant, Name, Effect and Type, where $D_{Effect}$ is the domain of the attribute Effect. This relation contains the opinions of experts and residents of a certain region about the effect of various pollution sources. An example of the relation's content is given in table \ref{survey}.

\begin{table}
	
	\begin{center}
		\begin{tabular}{|c|c|c|c|}
			\hline
			\textbf{Pollutant} & \textbf{Name} & \textbf{Effect} & \textbf{Type}\\
			\hline
			Oil      & A & Limited & Expert \\
			Oil      & B & Extreme & Resident \\
			Oil      & C & Moderate & Resident \\
			Oil      & D & Moderate & Expert \\
			Oil      & E & Tolerable & Resident \\
			Oil      & F & Severe & Resident \\
			Oil      & G & Tolerable & Expert \\
			Oil      & H & Limited & Expert \\
			Dioxin 	 & A & Severe & Expert \\
			Dioxin 	 & B & Irreversible & Resident \\
			Dioxin 	 & C & Major & Resident \\
			Dioxin 	 & D & Major & Expert \\
			Dioxin 	 & E & Severe & Resident \\
			Dioxin 	 & F & Extreme & Resident \\
			Dioxin 	 & G & Severe & Expert \\
			Dioxin 	 & H & Moderate & Expert \\
			Wastewater & A & Minimal & Resident \\
			Wastewater & C & Moderate & Resident \\
			Wastewater & D & Tolerable & Expert \\
			Wastewater & E & Limited & Resident \\
			Wastewater & F & Tolerable & Resident \\
			Wastewater & G & Limited & Expert \\
			Wastewater & H & Minimal & Expert \\

			\hline
		\end{tabular}
	\end{center}
	\caption{Table SURVEY}
	\label{survey}
\end{table}

Using this relation, we can make, for example, the following query: Which experts and residents are in a significant agreement about the ecological effects?
	
The execution of the query could unfold in the following way:

First the opinions of experts and residents are put into separate tables, by taking a projection of the table \ref{survey} over the attributes Pollutant, Name and Effect, with the appropriate restriction on the values of the Type attribute, where we take an arbitrary value for the similarity threshold of the Effect column (in this case it is 0.85). These projections will lead to tuple merging, so the resulting tables will contain multi-member sets as values of some columns.

$R_1=(\pi_{(Pollutant, Name, Effect)}(\sigma_{Type='Expert'})(SURVEY))$ with THRES(Effect)>0.85, THRES(Name)$\geq$0.0

$R_2=(\pi_{(Pollutant, Name, Effect)}(\sigma_{Type='Resident'})(SURVEY))$ with THRES(Effect)>0.85, THRES(Name)$\geq$0.0

The condition for the similarity threshold of the attribute Name is required to make the tuple merging possible. In case of leaving out the condition for similarity threshold of an attribute, its default value is 1. The Name attribute is crisp, i.e. the only possible values of the similarity relation are 0 and 1, so the only merging would occur for tuples having identical values in the Name column. The condition THRES(Name)$\geq$0.0 basically marks that attribute as irrelevant for the query, i.e. the tuples are merged depending on the values of the other attributes.

The values of the relations R1 and R2 are given in tables \ref{r1} and \ref{r2}.

\begin{table}
	
	\begin{center}
		\begin{tabular}{|c|c|c|}
			\hline
			\textbf{Pollutant} & \textbf{Name} & \textbf{Effect}\\
			\hline
			Oil      & $\{$A,D,G,H$\}$ & $\{$Limited, Moderate, Tolerable $\}$ \\
			Dioxin      & $\{$A,G $\}$ & Severe\\
			Dioxin      & D & Major\\
			Dioxin      & H & Moderate\\
			Wastewater      & $\{$A,D,G,H $\}$ & $\{$Minimal, Limited, Tolerable $\}$\\
			\hline
		\end{tabular}
	\end{center}
	\caption{Table R1}
	\label{r1}
\end{table}
\begin{table}
\begin{center}
	\begin{tabular}{|c|c|c|}
		\hline
		\textbf{Pollutant} & \textbf{Name} & \textbf{Effect}\\
		\hline
		Oil      & B & Extreme \\
		Oil     & $\{$C,E $\}$ & $\{$Moderate, Tolerable $\}$\\
		Oil      & F & Severe\\
		Dioxin      & B & Irreversible\\
		Dioxin      & $\{$C,F $\}$ & $\{$Major, Extreme $\}$\\
		Dioxin & E & Severe\\
		Wastewater & B & Severe\\
		Wastewater & $\{$C,E,F$\}$ & $\{$Limited, Tolerable, Moderate $\}$\\
		\hline
	\end{tabular}
\end{center}
\caption{Table R2}
\label{r2}
\end{table}

Finally, the result of the query is obtained by the natural join of the relations R1 and R2 on the values of the attributes Pollutant and Effect:

R3 = R1 JOIN R2 ON POLLUTANT,EFFECT WITH THRES(EFFECT)>0.85, THRES(NAME)>=0.0

The result of this operation is shown in table \ref{r3}.

\begin{table}
	\begin{center}
		\begin{tabular}{|c|c|c|c|}
			\hline
			\textbf{Pollutant} & \textbf{Name} & \textbf{Effect} & \textbf{Name}\\
			\hline
				Oil      & $\{$A,D,G,H$\}$ & $\{$Limited, Moderate, Tolerable $\}$ & $\{$C,E $\}$ \\ 
		
			Dioxin      & $\{$A,G $\}$ & Severe &F\\
			Dioxin      & D & $\{$Major, Extreme $\}$ & $\{$C,F $\}$\\
		
			Wastewater & $\{$A,D,G,H$\}$ & $\{$Minimal, Limited, Tolerable, Moderate $\}$ & $\{$C,E,F $\}$\\
			\hline
		\end{tabular}
	\end{center}
	\caption{Table R3}
	\label{r3}
\end{table}

Table R3 shows the lack of consensus about the effects of Dioxin, while there is much better agreement about oil and wastewater.

Note: The resulting table is copied verbatim from the original source. However, the relational model doesn't allow one table to contain two columns with the same name, so at least one column would have to be renamed during the natural join operation in the last query.
\end{pr}

\section{Proximity based model}

The validity of the key theorems in the Buckles-Petry model is made possible by the fact that similarity is an equivalence relation. However, in practice it is not always easy to define such a relation naturally, because of the strictness of the max-min transitivity condition. Shenoi and Melton cite a simple example: Let the domain $D$ contain the elements 1, 2 and 3 and assume we want to form a similarity relation over $D$. If we take $s(1,2)=0.8$ and $s(2,3)=0.8$, then the max-min transitivity requires that $s(1,3)$ be equal or greater than 0.8, which seems counter-intuitive.

Shenoi and Melton have extended the Buckles and Petry fuzzy relational model by substituting the similarity relation with the proximity relation. The proximity relation requires only the properties of reflexivity and symmetry and therefore represents a generalization of the similarity relation.

\begin{df} 
	The \textit{proximity relation} is a mapping $s_j:D_j\times D_j\rightarrow[0,1]$ such that for all $x,y\in D_j$:
\begin{enumerate}
  \item $s_j(x,x)=1,$ 
  \item $s_j(x,y)=s_j(y,x)$.
\end{enumerate}
\end{df}

The first step of the extension of the Buckles and Petry model is defining of the relations of $\alpha$-similarity and $\alpha$-proximity:

\begin{df} 
	Let $s_j$ be a proximity relation over the set $D_j$ and let there be given an arbitrary $\alpha\in[0,1]$. Then for the elements $x,y\in D_j$ we say that they are $\alpha$-\textit{similar}, denoted as $xS_\alpha y$, if and only if $s_j(x,y)\geq\alpha$.
	
\end{df}

\begin{df} 
	Let $s_j$ be a proximity relation over the set $D_j$ nad let an arbitrary $\alpha\in[0,1]$ be given. Then we say that the elements $x,z\in D_j$ are $\alpha$-\textit{proximate}, denoted as $xS_\alpha^+ z$ if and only if $xS_\alpha z$ or there exists a sequence of elements $y_1,y_2,\dots,y_r\in D_j$ such that $xS_\alpha y_1$, $y_1S_\alpha y_2$,\dots, $y_rS_\alpha z$.
\end{df}

It is easily shown that $S_\alpha^+$ is an equivalence relation and, therefore, it determines a set of partitions of the domain $D_j$. These partitions are identical to the sets produced in the Buckles and Petry model by repeating the tuple merging operation as long as it leads to valid tuples (a tuple is valid if all elements of each of its components are mutually similar to a degree equal to or greater than some given similarity threshold corresponding to that component's domain).

If $s$ is a similarity relation, then the relations $S_\alpha$ and $S_\alpha^+$ form equal equivalence classes.

All definitions in this domain are formulated for the so-called temporal domains, which represent the union of all attribute values currently appearing in the relation. Therefore, in this case, the content of the equivalence classes depends not only on the relation $s$ and domain $D_j$, but also on the current state of the database.

The definition of the relation $S_\alpha^+$, which ''compensates'' for the lack of the max-min transitivity, influences that some elements that are not mutually similar be drawn into the same equivalence class through some linking elements. The absence of some of these elements would mean breaking the chain and dividing the equivalence class. Similarly, the appearance of a new element can lead to the merging of some equivalence classes. Therefore, the insert and delete operations have a large influence on the appearance of the equivalence classes.

\begin{pr} \cite{shenoi1}
	A table containing the information on the physical characteristics of several persons is given. It includes information about the name, build and hair colour (table \ref{physcar}).
	
	\begin{table}
		\begin{center}
			\begin{tabular}{|c|c|c|}
				\hline
				\textbf{Name} & \textbf{Hair Color} & \textbf{Build}\\
				\hline
				Albert & Black & Large \\
				Bob & Dark brown & Average \\
				Charles & Auburn & Average \\
				David & Blond & Small \\
				Eugene & Dark brown & Average \\
				Frank & Red & Very small \\
				Gary & Bleached & Very large \\
				Henry & Blond & Average \\
				Ivan & Auburn & Large \\
				James & Blond & Large \\
				\hline
			\end{tabular}
		\end{center}
		\caption{Table PHYSICAL CHARACTERISTICS}
		\label{physcar}
	\end{table}

Over the domains of Hair Color and Build are defined the proximity relations shown in tables \ref{simhair} i \ref{simbuild}.

\begin{table}
	\begin{center}
		\begin{tabular}{|c|c|c|c|c|c|c|c|}
			\hline
			$s$ & Bk & DB & A & R & LB & Bd & Bc \\
			\hline
			Black (Bk) & 1.0 & 0.8 & 0.6 & 0.5 & 0.4 & 0.3 & 0.1 \\
			Dark brown (DB) & 0.8 & 1 & 0.7 & 0.6 & 0.6 & 0.5 & 0.2 \\
			Auburn (A) & 0.6 & 0.7 & 1 & 0.8 & 0.7 & 0.4 & 0.3 \\
			Red (R) & 0.5 & 0.6 & 0.8 & 1 & 0.7 & 0.5 & 0.4 \\
			Light brown (LB) & 0.4 & 0.6 & 0.7 & 0.7 & 1 & 0.7 & 0.5\\
			Blond (Bd) & 0.3 & 0.5 & 0.4 & 0.5 & 0.7 & 1 & 0.8 \\
			Bleached (Bc) & 0.1 & 0.2 & 0.3 & 0.4 & 0.5 & 0.8 & 1 \\
			\hline
		\end{tabular}
	\end{center}
	\caption{The hair color similarity table}
	\label{simhair}
\end{table}

\begin{table}
	
	\begin{center}
		\begin{tabular}{|c|c|c|c|c|c|}
			\hline
			$s$ & VL & L & A & S & VS \\
			\hline
			Very large (VL) & 1 & 0.7 & 0.5 & 0.3 & 0.1 \\
			Large (L) & 0.7 & 1 & 0.6 & 0.4 & 0.2 \\
			Average (A) & 0.5 & 0.6 & 1 & 0.6 & 0.4 \\
			Small (S) & 0.3 & 0.4 & 0.6 & 1 & 0.7 \\
			Very small (VS) & 0.1 & 0.2 & 0.4 & 0.7 & 1 \\
			\hline
		\end{tabular}
	\end{center}
	\caption{The build similarity table}
	\label{simbuild}
\end{table}

Assume that the persons in the database are suspects in an arson investigation. The culprit is suspected to have blond hair and large build, but due to the unreliability of the only witness it is necessary to allow certain deviations from the listed characteristics. One possible version of a fuzzy relational algebra query \cite{shenoi1} would be:

(Project (Select (PHYSICAL CHARACTERISTICS)
where HAIR COLOR = "Blond",
BUILD = "Large"
with LEVEL(HAIR COLOR) = 0.7,
LEVEL(BUILD) = 0.7)
with LEVEL(NAME) = 0.0,
LEVEL(HAIR COLOR) = 0.7,
LEVEL(BUILD) = 0.7
giving LIKELY ARSONISTS).

The result of the query is shown in the table \ref{suspects}.

\begin{table}
	
	\begin{center}
		\begin{tabular}{|c|c|c|}
			\hline
			\textbf{Name}&\textbf{Hair Color}&\textbf{Build}\\
			\hline
			$\{$Gary,James $\}$ & $\{$Blond, Bleached $\}$ & $\{$Very large, large $\}$\\
			\hline
		\end{tabular}
	\end{center}
	\caption{LIKELY ARSONISTS}
	\label{suspects}
\end{table}

In case that the database contains an additional person, say Kevin, with light brown hair and large build, the query result would consist of a single tuple containing the names of five suspects \cite{shenoi1}.
\end{pr}

\section{Using proximity relations to form equivalence classes on one-dimensional domains}

The Shenoi and Melton model, thanks to using proximity relations instead of similarity relations, provides to the users a much more natural way of defining the relationships between the domain elements and applicability to a much wider class of domains, particularly to a very important special case of the linearly ordered domains. However, this model is not without its drawbacks. It is natural, except possibly in some very special cases, that during the life of a database its tables grow, i.e. that their columns contain more and more elements of their respective domains. It is possible that at some point all or nearly all of the domain values be present in the database. Even when it is not the case, it is highly probable that some interval of characteristic values will be filled.

For example, in a database containing information about a large number of sports players, aside from the few most famous, the income of most others would be within some densely populated interval, where any two players will rarely have the exact same income, but there will be a large number of those whose salaries differ by very little. In cases such as this forming the equivalence classes in the way described in the previous chapter loses meaning, because  the starting value of $\alpha$ would have negligible influence on the appearance of the classes. It is easy to construct an example where such values as $\alpha=0.8$ and $\alpha=0.2$ lead to identical equivalence classes.

Another problem is related to the very meaning of the relation $S_\alpha^+$, which is not completely clear. 

This relation makes sense in cases where the proximity of the elements is defined by some time-dependent characteristic represented in the database. For example, given a set of cities, assume we define the proximity by the degree of ease and cost of arriving from one city to the other by using the available modes of transportation. Let's say that we can travel from the city A to the city B by plain, and then from B to C by bus. Even if the (geographical) distance between A and C is great, the existence of these transportation lines could in some context put them in the same equivalence class. However, assume that for some reason planes stop flying between A and B. It is quite probable that we can still go from A to C some other way, but now the journey is prolonged and more complicated, with an almost certain increase of the price. It is logical to imagine that in the first case A and C are more $\alpha$-proximate than in the second, i.e. that they belong to a same equivalence class for a higher $\alpha$ than in the second case.

However, such a relation doesn't give good results for domains in which the similarity of the elements is observed through their defining (and thus unchangeable) characteristics. The difference between black and blond hair is big because of the very nature of those colours, and it is counter-intuitive that those two different colours appear in the same equivalence class just because the database is filled up with various colours between them. In this case, it appears that dependence of the classes on the temporal domain represents an undesirable feature.

Therefore it seems desirable to formally define a way to form equivalence classes with a proximity relation, but without the domain dependence. In this chapter we will show a simple way of forming these classes in the case when the domain is a subset of the real line. The result is easily extended to any linearly ordered set, and even to some sets of linguistic terms (which often represent labels of fuzzy sets), which will later be illustrated by an example. In the next chapter we will apply a similar method to a two-dimensional set.

Let $I=[0,L]$, $L>0$, be an interval of real numbers. We will define a function 
$p:I\times I\rightarrow[0,1]$, \begin{equation}\label{proximityformula}
p(a,b)=1-\frac{|b-a|}{L}.
\end{equation} 

Obviously, $p(a,a)=1$ for each $a\in I$ and $p(a,b)=p(b,a)$ for each pair of elements $a$ and $b$ from $I$. We will call $p$ the \textit{proximity function} of the elements of the set $I$. In the Buckles-Petry model the similarity relation satisfies all properties of the definition of the equivalence relation, so it naturally partitions the domain into disjoint subsets whose union contains the set $I$. In is necessary to form equivalence classes by using a proximity function, namely without assuming the max-min transitivity property. More precisely, for a given value $\alpha\in(0,1]$ we will construct a family of sets $\{I_k\}$ such that $I_k\cap I_l=\emptyset$ for $k\neq l$ and $\cup_{k}I_k=I$, where $p(x,y)\geq\alpha$ for any two elements $x$ and $y$ that belong to the same set $I_k$.

From the definition of the function $p$ it follows that $p(a,b)\geq\alpha$ if and only if $|a-b|\leq(1-\alpha)L$. Let $m=(1-\alpha)L$ and $n=\lfloor\frac{1}{1-\alpha}\rfloor$. Let $I_k=[km,(k+1)m)$ for $k=0,1,\dots,n$.  The first few sets of this family are $I_0=[0,m)$, $I_1=[m,2m)$, $I_2=[2m,3m)$ etc. The sets $I_k$ are obviously disjoint, and we will also show that they cover the entire interval $I$, i.e. that $\cup_{k=0}^{n}I_k\supset I$. $n$ represents the number of intervals $I_k$ that have equal lengths. If $\frac{1}{1-\alpha}$ is an integer, then the number of the intervals $I_k$ is $n$, otherwise there are $n+1$ of them.

From the definition of the floor function follows the inequality $\lfloor\frac{1}{1-\alpha}\rfloor\leq\frac{1}{1-\alpha}\leq\lfloor\frac{1}{1-\alpha}\rfloor+1$. By multiplying it by $m$ we get $\lfloor\frac{1}{1-\alpha}\rfloor m\leq\frac{1}{1-\alpha}m\leq(\lfloor\frac{1}{1-\alpha}\rfloor+1)m$. Since $\lfloor\frac{1}{1-\alpha}\rfloor=n$ by definition, we get $nm\leq\frac{1}{1-\alpha}m\leq(n+1)m$. $m=(1-\alpha)L$, so the middle member will equal  $L$, i.e. $nm\leq L\leq(n+1)m$, so the right border of the interval $I$ is contained in the last set of the family, $I_n$. Therefore, instead of the interval $[nm,(n+1)m)$, by $I_n$ we will denote the set $[nm,(n+1)m)\cap I$, i.e. $I_n=[nm,L]$. Such a construction of the family of intervals $I_k$ enables that each element of the set $I$ belongs to precisely one of the sets $I_k$.

Now let $x$ and $y$ be members of the set $I$. We will define a relation $E$ on the set $I$ in the following way: $xEy$ if and only if there exists $k\in\overline{0,n}$ such that $x$ and $y$ belong to the same set $I_k$. It is easy to see that the following theorem holds.

\begin{teo} $E$ is an equivalence relation.
\end{teo}

Of course, the equivalence classes defined by $E$ are exactly the sets $I_k$.

\begin{teo} If $xEy$, then $p(x,y)\geq\alpha$.
\end{teo}
\begin{proof} 
	If $xEy$, then from the definition of $E$ follows that there exists a set $I_k=[km,(k+1)m)$ such that both $x$ and $y$ belong to $I_k$. The length of $I_k$ is $|I_k|=(k+1)m-km=m$, so $|x-y|\leq m$. Therefore, $|x-y|\leq(1-\alpha)L$. Now, $p(x,y)=1-\frac{|x-y|}{L}\geq 1-\frac{(1-\alpha)L}{L}=1-(1-\alpha)=\alpha$.
\end{proof}

Contrary to previous models \cite{buckles1,shenoi1}, in this case the converse does not hold, i.e. from the condition that the proximity of two elements is greater than some $\alpha$ it doesn't necessarily follow that those elements will belong to the same equivalence class formed for the given $\alpha$.

Depending on the choice of the similarity threshold $\alpha$, the size of the set $I_n$ could happen to be very small. For example, for $L=100$, $m=12$ and $n=8$ all the intervals $I_k$ will have length 12 except for $I_n$, whose length will be 4. This size difference cannot be avoided, except for the special values of $\alpha$ for which $\frac{1}{1-\alpha}$ is an integer.

The partition of any set into equivalence classes is unique for a given $\alpha$, so for each element we will always know the class to which it belongs. Therefore, it is possible to define a function that maps the pair (element,$\alpha$) to the class to which the element belongs for the given $\alpha$.

Let there be given the interval $I=[0,L]$ and $\alpha\in[0,1]$. Assume that for the given $\alpha$ $I$ is partitioned into $n+1$ classes, $I_1,I_2,\dots,I_{n+1}$, where $I_j=[(j-1)m,jm)$ for $j=1,2,\dots,n$ and $I_{n+1}=[nm,L]$. Any real $x\in I$ can be uniquely represented as $x=jm+r$, where $r<m$ \cite{koshy1}. By the definition of the class $I_j$ from the previous chapter, we see that $x$ belongs to the class $I_j$. Therefore, $x_1=j_1m+r_1$ and $x_2=j_2m+r_2$ will belong to the same equivalence class if and only if $j_1=j_2$.

\begin{pr}

The suppliers and parts database \cite{date1} contains for each supplier, among other records, the one about its status. Assume that valid status values are integers between 0 and 100. The table \ref{suppliers} (chapter \ref{sectionapplication}) contains in total 18 suppliers with their statuses, not all of which are distinct. The set of statuses that currently appear in the database is $S=\{10,20,25,30,35,40,45,50,55,60,65,75,80,90\}$. We will illustrate the content of the equivalence classes for a few different values of the proximity threshold $\alpha$:
\begin{enumerate}
  \item Let $\alpha=0.6$. Now $m=(1-0.6)\cdot 100=40$ and $n=\lfloor\frac{1}{1-0.6}\rfloor=2$. The interval $I=[0,100]$ will be divided in three parts, the first two of which are of the given length $m$: $I_1=[0,40),I_2=[40,80),I_3=[80,100)$. 
  All the elements from the same set $I_i,i=1,2,3$ are mutually proximate to a digree greater than or equal to 0.6. By $K_i$ we will denote the set $S\cap I_i$ (this is the equivalence class corresponding to the interval $I_i$). For this $\alpha$ value, the classes $K_i$ are as follows:
  $K_1=\{10,20,25,30,35\},$ $K_2=\{40,45,50,55,60,65,75\}$ and $K_3=\{80,90\}$.
  \item Let $\alpha=0.8$. Now $m=20$ and $n=5$, so all the intervals $I_i$ are of the same length: $I_1=[0,20),$ $I_2=[20,40),\dots,$ $I_5=[80,100)$. The equivalence classes are $K_1=\{10\},$ $K_2=\{20,25,30,35\},$ $K_3=\{40,45,50,55\},$ $K_4=\{60,65,75\}$ and $K_5=\{80,90\}.$
  \item For $\alpha=0.85$ we have $m=15,$ $n=6$ and equivalence classes $K_1=\{10\},$ $K_2=\{20,25\},$ $K_3=\{30,35,40\},$ $K_4=\{45,50,55\},$ $K_5=\{60,65\},$ $K_6=\{75,80\}$ and $K_7=\{90\}.$
  \item For $\alpha=0.92$, $m=8$ and we get $K_1=\{10\},$ $K_2=\{20\},$ $K_3=\{25,30\},$ $K_4=\{35\},$ $K_5=\{40,45\},$ $K_6=\{50,55\},$ $K_7=\{60\},$ $K_8=\{65\},$ $K_9=\{75\},$ $K_{10}=\{80\}$ i $K_{11}=\{90\}.$
\end{enumerate}

\end{pr}

\begin{pr}
Shenoi and Melton give the example of the domains of hair colours and build of arson suspects, together with the proximity relations over those domains \cite{shenoi1}. First we will show how by using the above method a similar proximity relation could be defined automatically. First we take the Build domain. It contains the following elements: very large (VL), large (L), average (A), small (S) and very small (VS). Obviously, on this domain we can define a linear order, which could not be well represented by a similarity relation. To the elements of this domain we will assign the following numerical values: VL - 0, L - 1, A - 2, S - 3, VS - 4. Now we have a set of five integers, between 0 and 4, so we let $L=4$. The relation $s$ is defined so that its value is not dependent on the particular elements, but only on their distance. Since, for example, the distance between VL and L is 1, the same as the distance between A and S, $s(VL,L)$ will be equal to $s(A,S)$. The similar holds for other element combinations. Therefore, instead of $s(a,b)$ we can write $s(|a-b|)$.

Using $\eqref{proximityformula}$ we get: $s(VL,L)=s(L,A)=s(A,S)=s(S,VS)=s(1)=1-\frac{1}{4}=0.75$. $s(2)=0.5,s(3)=0.25,s(4)=0$. In this way we get the values shown in the table \ref{sizesimilarity}.

\begin{table}

\begin{center}
\begin{tabular}{|c|c|c|c|c|c|}
  \hline
  $s$ & VL & L & A & S & VS \\
  \hline
  VL & 1 & 0.75 & 0.5 & 0.25 & 0 \\
  L & 0.75 & 1 & 0.75 & 0.5 & 0.25 \\
  A & 0.5 & 0.75 & 1 & 0.75 & 0.5 \\
  S & 0.25 & 0.5 & 0.75 & 1 & 0.75 \\
  VS & 0 & 0.25 & 0.5 & 0.75 & 1 \\
  \hline
\end{tabular}
\end{center}
\caption{Proximity of builds}
\label{sizesimilarity}
\end{table}

The other domain, Hair Colour, contains the following elements: black (Bk), dark brown (DB), auburn (A), red (R), light brown (LB), blond (Bd) and bleached (Bc). Analogous to the previous procedure, we get the values from the table \ref{hairsimilarity}. 

\begin{table}
\begin{center}
\begin{tabular}{|c|c|c|c|c|c|c|c|}
  \hline
  $s$ & Bk & DB & A & R & LB & Bd & Bc \\
  \hline
  Bk & 1 & 0.83 & 0.67 & 0.5 & 0.33 & 0.16 & 0 \\
  DB & 0.83 & 1 & 0.83 & 0.67 & 0.5 & 0.33 & 0.16 \\
  A & 0.67 & 0.83 & 1 & 0.83 & 0.67 & 0.5 & 0.33 \\
  R & 0.5 & 0.67 & 0.83 & 1 & 0.83 & 0.67 & 0.5 \\
  LB & 0.33 & 0.5 & 0.67 & 0.83 & 1 & 0.83 & 0.67\\
  Bd & 0.16 & 0.33 & 0.5 & 0.67 & 0.83 & 1 & 0.83 \\
  Bc & 0 & 0.16 & 0.33 & 0.5 & 0.67 & 0.83 & 1 \\
  \hline
\end{tabular}
\end{center}
\caption{Proximity of hair colours}
\label{hairsimilarity}
\end{table}

These values don't differ much from the ones chosen by Shenoi and Melton.

Now let $\alpha=0.8$ be a given proximity threshold. We will form the corresponding equivalence classes for this domain. To the element Bk we will assign the value 0, to the element DB 1 and so on, to Bc=6. The interval to which these elements belong is $I=[0,6]$ and its length is 6. $(1-0.8)\times 6=1.2$, so $m=1.2$ and the interval $I$ will be divided into the subintervals $I_1=[0,1.2)$, $I_2=[1.2,2.4)$, $I_3=[2.4,3.6)$, $I_4=[3.6,4.8)$ and $I_5=[4.8,6]$. Now $Bk,DB\in I_1$, $A\in I_2$, $R\in I_3$, $LB\in I_4$ and $Bd,Bc\in I_5$, so we get the equivalence classes $\{Bk,DB\}$, $\{A\}$, $\{R\}$, $\{LB\}$ and $\{Bd,Bc\}$.

By diminishing the value $\alpha$ we get a smaller number of classes with larger number of members. For example, for $\alpha=0.6$ the value $m$ will be 2.4, and the corresponding equivalence classes are $\{Bk,DB\}$, $\{A,R\}$ and $\{LB,Bd,Bc\}$. For $\alpha=0.5$ we get $m=3$ and the classes $\{Bk,DB,A\}$, $\{R,LB\}$ and $\{Bd,Bc\}$. For $\alpha=0.3$ the equivalence classes are $\{Bk,DB,A\}$ and $\{R,LB,Bd,Bc\}$. These examples show that, unlike the Shenoi and Melton model, the equivalence classes corresponding to a value of $\alpha$ do not necessarily have to be subsets of the classes corresponding to a larger value of $\alpha$. Some elements belonging to a same class can be split by increasing $\alpha$. Increasing the domain cardinality reduces the prominence of this occurrence.

\end{pr}

\subsection{Equalization of equivalence classes sizes}\label{equalization}

In the previous part of the text, the intervals that determine the equivalence classes were formed as follows: For a given $\alpha\in[0,1]$, the length of the interval, $m$, was given by the formula $m=(1-\alpha)L$, where $L$ is the length of the entire domain. This partition results in $n$ classes of the length $m$ and, sometimes, one additional, $n+1$-st, class of a smaller length, where $n=\lfloor\frac{1}{1-\alpha}\rfloor$. Only in the case when $\frac{1}{1-\alpha}$ is an integer the number of classes will be precisely $n$. We can change the definition of $m$ to $m=\frac{L}{n}$. In this case, all classes will be of the same length for each $\alpha\in[0,1]$, but that length and the number of classes will be identical for all values of $\alpha$ in a certain interval. For example, if $\alpha=0$, we get only one class, equal to the entire domain. $\alpha=0.5$ results in two equal classes, $\alpha=\frac{2}{3}=0.666...$ in three, etc. Also, we get two equal classes for each $\alpha$ between 0 and 0.5, three for each $\alpha$ between 0.5 and 0.666... etc. By increasing precision, i.e. by letting $\alpha$ approach 1, we get smaller and smaller intervals, i.e. the significance of the choice of $\alpha$ increases.

\section{Forming equivalence classes on two-dimensional domains by using proximity relations}

Let there be given a real number $L>0$ and the square $K=[0,L]\times[0,L]\subset\mathbb{R}^2$. For the points $M_1(x_1,y_1)$ and $M_2(x_2,y_2)$, by $d(M_1,M_2)$ we will denote their Euclidean distance: $d(M_1,M_2)=\sqrt{(x_2-x_1)^2+(y_2-y_1)^2}$. The proximity function of the points $M_1$ and $M_2$ is defined by $p(M_1,M_2)=1-\frac{d(M_1,M_2)}{\sqrt{2}L}$. This definition implies that $p(M_1,M_2)\geq\alpha$ if and only if $d(M_1,M_2)\leq(1-\alpha)\sqrt{2}L$.

The equivalence classes over the square $K$ depending on $\alpha\in[0,1]$, such that the proximity of all elements in a class is greater than $\alpha$ can be formed similarly as in the one-dimensional case. We will denote by $m$ the number $(1-\alpha)L$ and partition the interval $[0,L]$ as previously described. Two points will be in the same equivalence class (i.e. in the same square) if their projections on the coordinate axes are in the same equivalence classes defined on those axes.

More precisely: Let $\alpha\in[0,1]$ be given. Now $m=(1-\alpha)L$ and $n=\lfloor\frac{1}{1-\alpha} \rfloor$. In general, the interval $[0,L]$ will be partitioned into $n+1$ subintervals: $I_1=[0,m)$, $I_2=[m,2m),\dots$, $I_n=[(n-1)m,nm)$ and $I_{n+1}=[nm,L]$. The Cartesian products of these intervals are squares $I_{ij}=I_i\times I_j=[(i-1)m,im)\times[(j-1)m,jm)$, for $i,j=0,1,2,\dots,n$. For $i=n+1$ and $j\leq n$ we have $I_{n+1,j}=[nm,L]\times[(j-1)m,jm)$ and similar for the square $I_{i,n+1}$ for $i\leq n$. The last square, with the largest values $x$ and $y$ is $I_{n+1,n+1}=[nm,L]\times[nm,L]$.

\begin{df} Let $\alpha\in[0,1]$ be given. We say that the points $M_1$ and $M_2$ are in the relation $E_\alpha$ if and only if in the partition of the square $[0,L]\times[0,L]$ with respect to the given $\alpha$ they belong to the same square $I_{ij}$.
\end{df}

The relation $E_\alpha$ is obviously an equivalence relation since the squares $I_{ij}$ are disjoint and they cover the entire square $K$. What remains is to show that the points in any square are mutually proximate to a degree larger than or equal to $\alpha$.

\begin{teo} Let $\alpha\in[0,1]$. If the points $M_1$ and $M_2$ belong to the same square $I_{ij}$, then $p(M_1,M_2)\geq\alpha$.
\end{teo}

\begin{proof} Let the points $M_1(x_1,y_1)$ and $M_2(x_2,y_2)$ be given. If $M_1$ and $M_2$ are in the same square $I_{ij}$ for a given $\alpha$, then $|x_2-x_1|\leq m$ and $|y_2-y_1|\leq m$, since the sides of the square $I_{ij}$ are of the length $m$. Now $d(M_1,M_2)=\sqrt{(x_2-x_1)^2+(y_2-y_1)^2}\leq\sqrt{m^2+m^2}=\sqrt{2m^2}=\sqrt{2}m=\sqrt{2}(1-\alpha)L$, so $p(M_1,M_2)=1-\frac{d(M_1,M_2)}{\sqrt{2}L}\geq 1-\frac{\sqrt{2}(1-\alpha)L}{\sqrt{2}L}=1-(1-\alpha)=\alpha$.
\end{proof}

\begin{pr}
The suppliers and parts database contains a table with the locations of all cities. The locations are given as $X$ and $Y$ coordinates in the square $K=[0,100]\times[0,100]$ (table \ref{citycoordinates}).
\begin{table}
\begin{center}
\begin{tabular}{|c|c|c|}
  \hline
  CITY & X & Y\\
  \hline
  Shire & 12 & 7\\
  Bree & 14 & 23\\
  Rivendell & 23 & 28\\
  Isengard & 45 & 48\\
  Moria & 45 & 70\\
  Gondor & 63 & 26\\
  Rohan & 68 & 22\\
  Lotlorien & 82 & 43\\
  Mordor & 91 & 82\\
  \hline
\end{tabular}
\end{center}

\caption{City locations}
\label{citycoordinates}
\end{table}

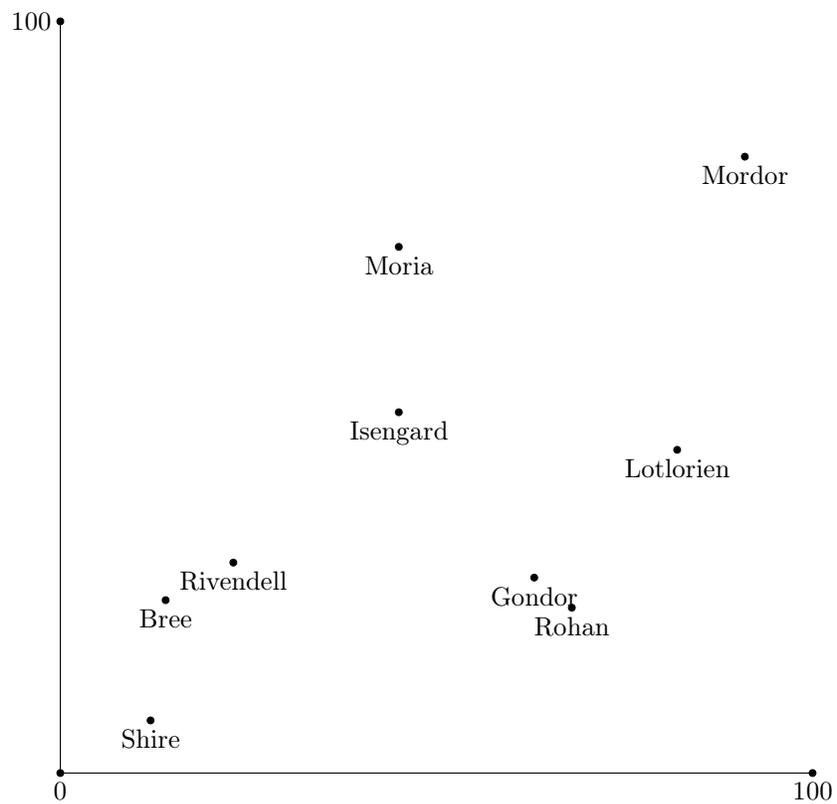
\begin{figure}
	\centering
\begin{tikzpicture} [scale=0.1]
\draw (0,0) -- (100,0);
\draw (0,0) -- (0,100);
\draw[fill] (12,7) circle [radius=0.45];
\node [below] at (12,7) {Shire};
\draw[fill] (14,23) circle [radius=0.45];
\node [below] at (14,23) {Bree};
\draw[fill] (23,28) circle [radius=0.45];
\node [below] at (23,28) {Rivendell};
\draw[fill] (45,48) circle [radius=0.45];
\node [below] at (45,48) {Isengard};
\draw[fill] (45,70) circle [radius=0.45];
\node [below] at (45,70) {Moria};
\draw[fill] (63,26) circle [radius=0.45];
\node [below] at (63,26) {Gondor};
\draw[fill] (68,22) circle [radius=0.45];
\node [below] at (68,22) {Rohan};
\draw[fill] (82,43) circle [radius=0.45];
\node [below] at (82,43) {Lotlorien};
\draw[fill] (91,82) circle [radius=0.45];
\node [below] at (91,82) {Mordor};
\draw[fill] (0,0) circle [radius=0.45];
\node [below] at (0,0) {0};
\draw[fill] (100,0) circle [radius=0.45];
\node [below] at (100,0) {100};
\draw[fill] (0,100) circle [radius=0.45];
\node [left] at (0,100) {100};
\end{tikzpicture}
\caption{City locations}
\label{citylocations}
\end{figure}

Let $\alpha=0.8$. Now $m=(1-\alpha)\cdot 100=20$ and $n=\lfloor\frac{1}{1-\alpha}\rfloor=\lfloor 5\rfloor=5$. In this case, $\frac{1}{1-\alpha}$ is an integer, so the number of equivalence classes on one coordinate axis will be $n$, where the last class is $I_n=[(n-1)m,nm]$, i.e. in this case $I_5=[80,100]$. The other classes are $I_1=[0,20),$ $I_2=[20,40)$, $I_3=[40,60)$ and $I_4=[60,80)$ 

The square $K=[0,100]\times[0,100]$ is partitioned into 25 squares with the size of the length 20. $K_{11}=I_1\times I_1=[0,20)\times[0,20)$, $K_{12}=[0,20)\times[20,40)$, $\dots$, $K_{ij}=[20(i-1),20i]\times[20(j-1),20j]...$. 

The cities in the database belong to the following equivalence classes: $Shire\in K_{11}, Bree\in K_{12}, Rivendel\in K_{22},Isengard\in K_{33}, Moria\in K_{34}, Gondor\in K_{42}, Rohan\in K_{42}, Lotlorien\in K_{43}, Mordor\in K_{55}$. Therefore, all cities are in separate equivalence classes, except for Gondor and Rohan (figure \ref{citiesalpha08})

\begin{figure}
	\centering
	\begin{tikzpicture} [scale=0.1]
	\draw[help lines, step=20] (0,0) grid (100,100);
	\draw[fill] (12,7) circle [radius=0.45];
	\node [below] at (12,7) {Shire};
	\draw[fill] (14,23) circle [radius=0.45];
	\node [below] at (14,23) {Bree};
	\draw[fill] (23,28) circle [radius=0.45];
	\node [below] at (23,28) {Rivendell};
	\draw[fill] (45,48) circle [radius=0.45];
	\node [below] at (45,48) {Isengard};
	\draw[fill] (45,70) circle [radius=0.45];
	\node [below] at (45,70) {Moria};
	\draw[fill] (63,26) circle [radius=0.45];
	\node [below] at (63,26) {Gondor};
	\draw[fill] (68,22) circle [radius=0.45];
	\node [below] at (68,22) {Rohan};
	\draw[fill] (82,43) circle [radius=0.45];
	\node [below] at (82,43) {Lotlorien};
	\draw[fill] (91,82) circle [radius=0.45];
	\node [below] at (91,82) {Mordor};
	\draw (0,0) -- (100,0);
	\draw (0,0) -- (0,100);
	\draw[fill] (0,0) circle [radius=0.45];
	\node [below] at (0,0) {0};
	\draw[fill] (100,0) circle [radius=0.45];
	\node [below] at (100,0) {100};
	\draw[fill] (0,100) circle [radius=0.45];
	\node [left] at (0,100) {100};
	\draw[fill] (0,20) circle [radius=0.45];
	\node [left] at (0,20) {20};
	\draw[fill] (0,40) circle [radius=0.45];
	\node [left] at (0,40) {40};
	\draw[fill] (0,60) circle [radius=0.45];
	\node [left] at (0,60) {60};
	\draw[fill] (0,80) circle [radius=0.45];
	\node [left] at (0,80) {80};
	\draw[fill] (20,0) circle [radius=0.45];
	\node [below] at (20,0) {20};
	\draw[fill] (40,0) circle [radius=0.45];
	\node [below] at (40,0) {40};
	\draw[fill] (60,0) circle [radius=0.45];
	\node [below] at (60,0) {60};
	\draw[fill] (80,0) circle [radius=0.45];
	\node [below] at (80,0) {80};
	\end{tikzpicture}
	\caption{Equivalence classes for $\alpha=0.8$}
	\label{citiesalpha08}
\end{figure}
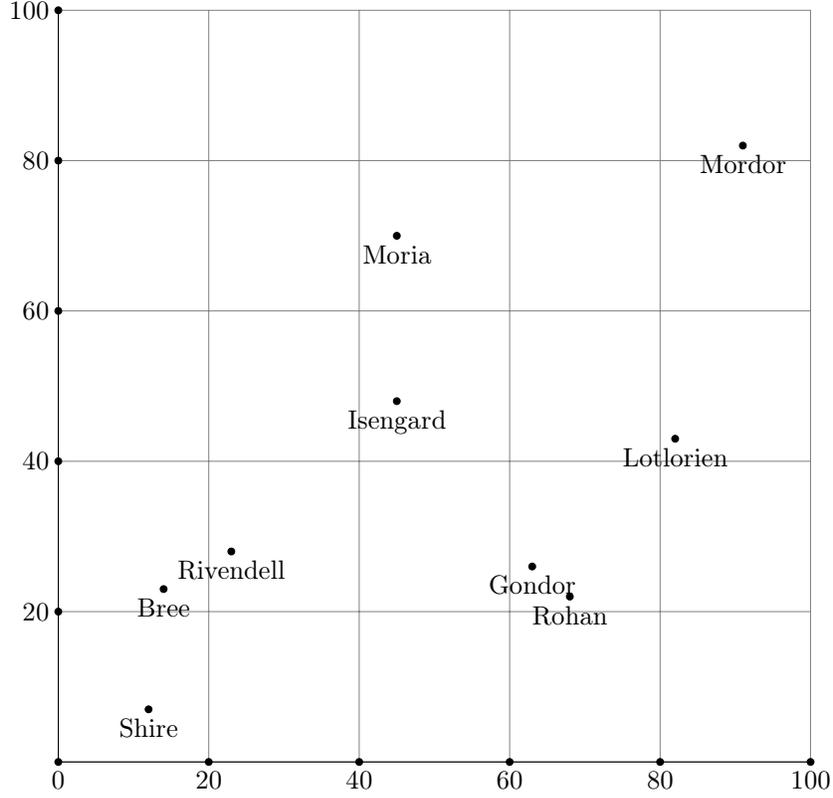

Now let $\alpha=0.6$. $m=(1-\alpha)\cdot 100=40$, $n=2$, so on each coordinate axes we get three equivalence classes: $I_1=[0,40)$, $I_2=[40,80)$, $I_3=[80,100)$. The equivalence classes on the square $K$ are $I_{11}=[0,40)\times[0,40)$, $I_{12}=[0,40)\times[40,80)$, $I_{13}=[0,40)\times[80,100]$, $I_{21}=[40,80)\times[0,40)$, $I_{22}=[40,80)\times[40,80)$, $I_{23}=[40,80)\times[80,100]$, $I_{31}=[80,100]\times[0,40)$, $I_{32}=[80,100]\times[40,80)$ and $I_{33}=[80,100]\times[80,100]$.

The layout of the cities in the equivalence classes is as follows: $Shire, Bree, Rivendel\in K_{11}, Isengard, Moria\in K_{22}, Gondor, Rohan\in K_{21}, Lotlorien\in K_{32}, Mordor\in K_{33}$ (figure \ref{citiesalpha06}).

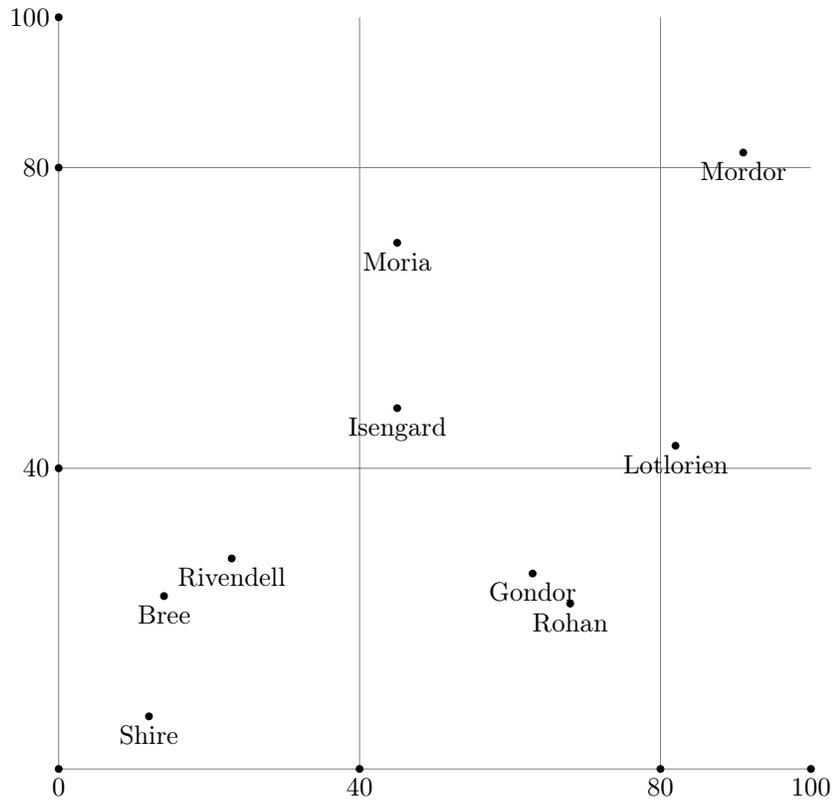
\begin{figure}
	\begin{tikzpicture} [scale=0.1]
	\draw[help lines, step=40] (0,0) grid (100,100);
	\draw[fill] (12,7) circle [radius=0.45];
	\node [below] at (12,7) {Shire};
	\draw[fill] (14,23) circle [radius=0.45];
	\node [below] at (14,23) {Bree};
	\draw[fill] (23,28) circle [radius=0.45];
	\node [below] at (23,28) {Rivendell};
	\draw[fill] (45,48) circle [radius=0.45];
	\node [below] at (45,48) {Isengard};
	\draw[fill] (45,70) circle [radius=0.45];
	\node [below] at (45,70) {Moria};
	\draw[fill] (63,26) circle [radius=0.45];
	\node [below] at (63,26) {Gondor};
	\draw[fill] (68,22) circle [radius=0.45];
	\node [below] at (68,22) {Rohan};
	\draw[fill] (82,43) circle [radius=0.45];
	\node [below] at (82,43) {Lotlorien};
	\draw[fill] (91,82) circle [radius=0.45];
	\node [below] at (91,82) {Mordor};
	\draw[fill] (0,0) circle [radius=0.45];
	\node [below] at (0,0) {0};
	\draw[fill] (100,0) circle [radius=0.45];
	\node [below] at (100,0) {100};
	\draw[fill] (0,100) circle [radius=0.45];
	\node [left] at (0,100) {100};
	\draw[fill] (0,40) circle [radius=0.45];
	\node [left] at (0,40) {40};
	\draw[fill] (40,0) circle [radius=0.45];
	\node [below] at (40,0) {40};
	\draw[fill] (80,0) circle [radius=0.45];
	\node [below] at (80,0) {80};
	\draw[fill] (0,80) circle [radius=0.45];
	\node [left] at (0,80) {80};
	\end{tikzpicture}
	\caption{Equivalence classes for $\alpha=0.6$}
	\label{citiesalpha06}
\end{figure}

\end{pr}
\section{Application}\label{sectionapplication}
The database contains a table with names, statuses and cities of several suppliers (table \ref{suppliers}).

\begin{table}
\begin{center}
\begin{tabular}{|c|c|c|}
  \hline
  SNAME & STATUS & CITY\\
  \hline
  Bagins & 20 & Shire\\
  Took & 55 & Shire\\
  Proudfoot & 30 & Shire\\
  Sauron & 80 & Mordor\\
  Elrond & 75 & Rivendell\\
  Arwen & 60 & Rivendell\\
  Glorfindel & 65 & Rivendell\\
  Eomer & 50 & Rohan\\
  Eowyn & 40 & Rohan\\
  Theoden & 25 & Rohan\\
  Grima & 35 & Gondor\\
  Gamgee & 60 & Shire\\
  Galadriel & 75 & Lothlorien\\
  Gimli & 80 & Moria\\
  Saruman & 90 & Isengard\\
  Balrog & 55 & Moria\\
  Denethor & 10 & Gondor\\
  Barliman & 45 & Bree\\
  \hline
\end{tabular}
\end{center}
\caption{Suppliers}
\label{suppliers}
\end{table}

When defining the proximity functions for attribute domains, one should also pay attention to the role the attribute plays in a query. The attribute SNAME in this case is the primary key and it will most often feature in query results, while the additional attributes STATUS and CITY will be used for restriction conditions. A typical example of a crisp query on this table could be ''Find the names of all suppliers from Mordor with status greater than 20''. Under these conditions, for proximity relations over the attributes STATUS and CITY we can keep the relations shown in the examples from the previous chapters, while for the SNAME attribute the proximity relation would look like this: $p(x,y)=1$ for all $x$ and $y$ and any $\alpha<1$. Since the attribute SNAME is the ''target'' of the query, i.e. we don't perform any comparisons of its values, it is natural that we disregard all differences between them, namely, for the output purposes we consider them all to be equal. This way, the content of the tuples produced by the merge operation depends only on the equivalence classes over the other two attributes.

In case that the query contains a comparison over the SNAME attribute (for example ''List all suppliers whose names begin with some of the first letters of the alphabet''), we would obviously have to define another proximity function. The similar would be with the attribute CITY if, for example, we compare the cities by names and not by position: ''List the names of all the suppliers from the cities whose names are similar to Rohan''. This example shows that, depending on the context, various different proximity functions could exist over a single attribute.

Assume that we don't make a difference in the names of the suppliers and that the proximity functions over the two remaining attributes are as previously defined. By merging the tuples of the previous relation for $\alpha=0.6$ we get the relation shown in the table \ref{classesalpha06}.

\begin{table}
\begin{center}
\begin{tabular}{|c|c|c|}
  \hline
  SNAME & STATUS & CITY\\
  \hline
  Bagins, Proudfoot & 20,30 & Shire\\
  Took, Arwen, Glorfindel,  & \multirow{2}*{45,55,60,65} & \multirow{2}*{Shire, Rivendell,Bree}\\
  Gamgee, Barliman& &\\
  Sauron & 80 & Mordor\\
  Elrond & 75 & Rivendell\\
  Eomer, Eowyn & 40,50 & Rohan\\
  Denethor, Theoden, Grima & 10,25,35 & Rohan, Gondor\\
  Galadriel & 75 & Lothlorien\\
  Gimli, Saruman & 80,90 & Moria, Isengard\\
  Balrog & 55 & Moria\\

  \hline
\end{tabular}
\end{center}
\caption{Equivalence classes for $\alpha=0.6$}
\label{classesalpha06}
\end{table}

Informally, the second tuple in the resulting relation would represent the suppliers having a medium status and located in the southwest of the map.

If, for example, we are interested only in the suppliers' location and not in their status, we would perform a projection of this relation onto the attributes SNAME and CITY, which would lead to an additional tuple merging. For example, Balrog would join Gimli and Saruman (which in the previous table was prevented by their lower status) and Eomer and Eowyn to Denethor, Theoden and Grima:
$K_1=\{10,20,25,30,35\},$ $K_2=\{40,45,50,55,60,65,75\}$ and $K_3=\{80,90\}$.
$Shire, Bree, Rivendel\in K_{11}, Isengard, Moria\in K_{22}, Gondor, Rohan\in K_{21}, Lotlorien\in K_{32}, Mordor\in K_{33}$ (table \ref{projection}).

\begin{table}
\begin{center}
\begin{tabular}{|c|c|}
  \hline
  SNAME & CITY\\
  \hline
   Bagins, Proudfoot,Took, Arwen, Glorfindel, & \multirow{2}*{Shire, Rivendell, Bree}\\
    Gamgee, Barliman, Elrond&\\
  Sauron & Mordor\\
  Denethor, Theoden, Grima, Eomer, Eowyn & Rohan, Gondor\\
  Galadriel & Lothlorien\\
  Gimli, Saruman, Balrog  & Moria, Isengard\\

  \hline
\end{tabular}
\end{center}
\caption{Projection over SNAME and CITY}
\label{projection}
\end{table}

\section{Comparison of methods}

In this section we will show the results of both approaches to forming equivalence classes on a joint example. Table \ref{cities} contains 23 cities in Great Britain located between longitudes 0 and 2 west and latitudes 51 and 53 north (figure \ref{cities-mapa}). By $X$ and $Y$ we denote the modified coordinates, obtained by mapping the real ones to the set $[0,2]\times[0,2]$.

\begin{table}
	\begin{center}
		\begin{tabular}{|c|c|c|c|c|c|}
			\hline
			\textbf{CITY} & \textbf{No}& \textbf{WL} & \textbf{X} & \textbf{NL} & \textbf{Y}\\
			\hline
			Peterborough	&1&0.2508&	1.7492&52.5739& 	1.5739\\		
			Solihull		&2&1.7782&0.2218&52.4128&		1.4128\\		
			Oxford			&3&1.2577&0.7423	&51.7520&	0.7520\\		
			London		&4&0.1181&1.8819&51.5099&		0.5099\\		
			Swindon		&5&1.7722&0.2278		&51.5685&0.5685\\		
			Northampton		&6&0.9027&1.0973&	52.2405&	1.2405\\		
			Rugby			&7&1.2650&0.7350	&52.3709&	1.3709\\		
			Sutton Coldfield&8	&1.8240&0.1760&52.5704&		1.5704\\		
			Salisbury		&9&1.7945&0.2055		&51.0688&0.0688\\		
			Bedford		&10&0.4607&1.5393	&52.1364&	1.1364\\		
			Frankton	&11	&1.3776&0.6224	&52.3284&	1.3284\\		
			Coventry	&12	&1.5197&0.4903	&52.4068&	1.4068\\		
			Slough		&13	&0.5950&1.4050	&51.5105&	0.5105\\		
			Esher		&14	&0.3659&1.6341	&51.3695&	0.3695\\		
			Epsom		&15	&0.2674&1.7326	&51.3360&	0.3360\\		
			Borehamwood	&16	&0.2723&1.7277&51.6577		&0.6577\\		
			Crawley		&17&0.1872&1.8128	&51.1091&	0.1091\\				
			St. Neots	&18	&0.2651&1.7349	&52.2301&	1.2301\\		
			Meriden		&19&1.6478&0.3522	&52.4387&	1.4387\\
			Nottingham 	&20	&1.1581&0.8419&52.9548		&1.9548\\	
			Derby		&21	&1.4746&0.5254&52.9225		&1.9225\\
			Reading		&22&0.9781&1.0219&51.4543		&0.4543\\	
			Berkshire	&23	&1.1854&0.8146&51.4670		&0.4670\\	
			\hline
		\end{tabular}
	\end{center}
\caption{List of cities and coordinates}
\label{cities}
\end{table}

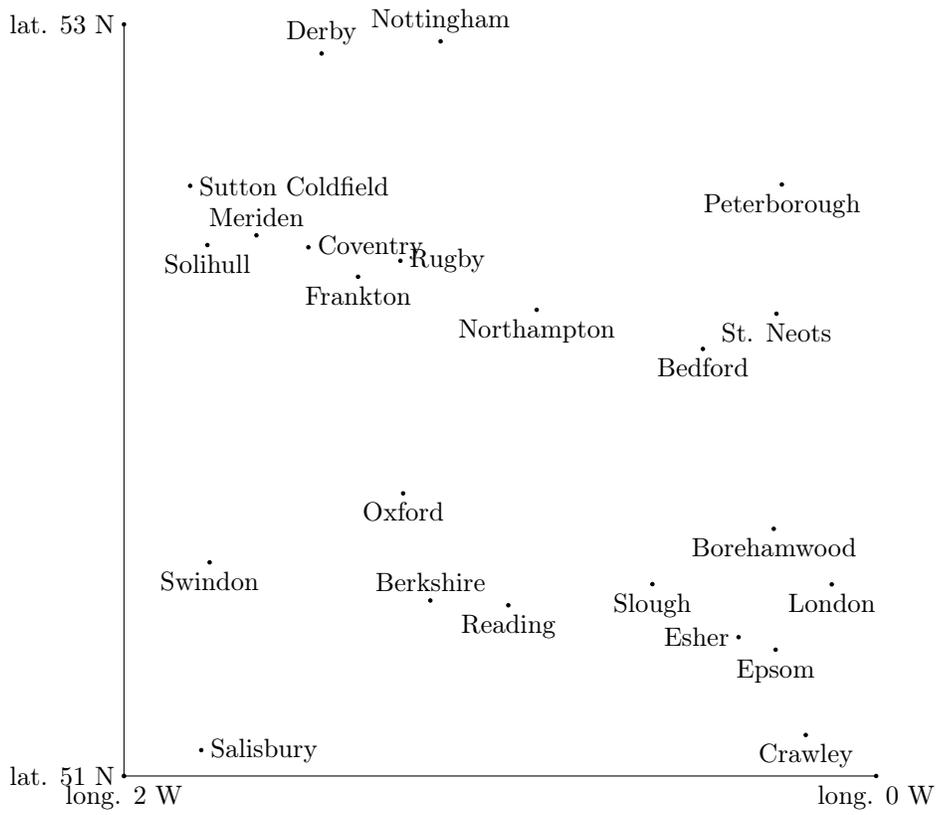
\begin{figure}
	\centering
	\begin{tikzpicture} [scale=5]
	\draw[fill] (1.7492,1.5739) circle [radius=0.0045];
	\node [below] at (1.7492,1.5739) {Peterborough};
	\draw[fill] (0.2218,1.4128) circle [radius=0.0045];
	\node [below] at (0.2218,1.4128) {Solihull};
	\draw[fill] (0.7423,0.7520) circle [radius=0.0045];
	\node [below] at (0.7423,0.7520) {Oxford};
	\draw[fill] (1.8819,0.5099) circle [radius=0.0045];
	\node [below] at (1.8819,0.5099) {London};
	\draw[fill] (0.2278,0.5685) circle [radius=0.0045];
	\node [below] at (0.2278,0.5685) {Swindon};
	\draw[fill] (1.0973,1.2405) circle [radius=0.0045];
	\node [below] at (1.0973,1.2405) {Northampton};
	\draw[fill] (0.1760,1.5704) circle [radius=0.0045];
	\node [right] at (0.1760,1.5704) {Sutton Coldfield};
	\draw[fill] (0.2055,0.0688) circle [radius=0.0045];
	\node [right] at (0.2055,0.0688) {Salisbury};
	\draw[fill] (1.5393,1.1364) circle [radius=0.0045];
	\node [below] at (1.5393,1.1364) {Bedford};
	\draw[fill] (0.6224,1.3284) circle [radius=0.0045];
	\node [below] at (0.6224,1.3284) {Frankton};
	\draw[fill] (0.4903,1.4068) circle [radius=0.0045];
	\node [right] at (0.4903,1.4068) {Coventry};
	\draw[fill] (1.4050,0.5105) circle [radius=0.0045];
	\node [below] at (1.4050,0.5105) {Slough};
	\draw[fill] (1.6341,0.3695) circle [radius=0.0045];
	\node [left] at (1.6341,0.3695) {Esher};
	\draw[fill] (1.7326,0.3360) circle [radius=0.0045];
	\node [below] at (1.7326,0.3360) {Epsom};
	\draw[fill] (1.7277,0.6577) circle [radius=0.0045];
	\node [below] at (1.7277,0.6577) {Borehamwood};
	\draw[fill] (1.8128,0.1091) circle [radius=0.0045];
	\node [below] at (1.8128,0.1091) {Crawley};
	\draw[fill] (1.7349,1.2301) circle [radius=0.0045];
	\node [below] at (1.7349,1.2301) {St. Neots};
	\draw[fill] (0.3522,1.4387) circle [radius=0.0045];
	\node [above] at (0.3522,1.4387) {Meriden};
	\draw[fill] (0.8419,1.9548) circle [radius=0.0045];
	\node [above] at (0.8419,1.9548) {Nottingham};
	\draw[fill] (0.5254,1.9225) circle [radius=0.0045];
	\node [above] at (0.5254,1.9225) {Derby};
	\draw[fill] (1.0219,0.4543) circle [radius=0.0045];
	\node [below] at (1.0219,0.4543) {Reading};
	\draw[fill] (0.8146,0.4670) circle [radius=0.0045];
	\node [above] at (0.8146,0.4670) {Berkshire};
	\draw[fill] (0.7350,1.3709) circle [radius=0.0045];
	\node [right] at (0.7350,1.3709) {Rugby};
	\draw (0,0) -- (2,0);
	\draw (0,0) -- (0,2);
	\draw[fill] (0,0) circle [radius=0.0045];
	\node [below] at (0,0) {long. 2 W};
	\draw[fill] (2,0) circle [radius=0.0045];
	\node [below] at (2,0) {long. 0 W};
	\draw[fill] (0,0) circle [radius=0.0045];

	\node [left] at (0,0) {lat. 51 N};
	\draw[fill] (0,2) circle [radius=0.0045];
	\node [left] at (0,2) {lat. 53 N};
	
	\end{tikzpicture}
	\caption{Region map}
	\label{cities-mapa}
\end{figure}

We will represent the map as a square with side length 2. Let $\alpha=0.4$. By using the formulas from section 5, we get $m=(1-0.4)\cdot 2=1.2$ and $n=\lfloor\frac{1}{0.6}\rfloor=\lfloor 1.67\rfloor=1$. In addition to one ''normal'' class, which we will denote by $I_{11}$, we get three additional, smaller, classes: ($I_{12}$, $I_{21}$ i $I_{22}$). The class $I_{11}$ contains the cities with coordinates between longitude 0.8 and 2 west and latitude 51 and 52.2 north. In the class $I_{12}$ are the cities with coordinates from 0.8 to 2 west longitude nad 52.2 and 53 north latitude and the classes $I_{21}$ and $I_{22}$ fill the remaining space.

Having in mind the city coordinates and the dimensions and positions of squares $I_{11}$ and $I_{12}$ and rectangles $I_{21}$ and $I_{12}$, we get the following equivalence classes:

$K_{11}=\{\text{Oxford, Swindon, Salisbury, Reading, Berkshire} \}$

$K_{12} = \{$Solihull, Northampton, Rugby, Sutton Coldfield, Frankton, Coventry, Meriden, Nottingham, Derby $\}$

$K_{21}=\{$London, Bedford, Slough, Esher, Epsom, Borehamwood, Crawley $\}$

$K_{22}=\{$Peterborough, St. Neots $\}$.

For $\alpha=0.6$ we get 9 classes, 4 of which have width and length 0.8 and the smallest one is of width and length 0.4 (the remaining 4 are rectangles with sides of different lengths)
The equivalence classes are as follows:

$K_{11}=\{$Oxford, Swindon, Salisbury$\}$

$K_{12}=\{$Solihull, Rugby, Sutton Coldfield, Frankton, Coventry, Meriden$\}$

$K_{13}=\{$Derby$\}$

$K_{21}=\{$Slough, Reading, Berkshire$\}$

$K_{22}=\{$Northampton, Bedford$\}$

$K_{23}=\{$Nottingham$\}$

$K_{31}=\{$London, Esher, Epsom, Borehamwood, Crawley $\}$

$K_{32}=\{$Peterborough, St. Neots$\}$.

The class $K_{33}$ is empty.

For $\alpha=0.8$ we get 25 classes of equal dimensions (the length of the side of the square is 0.4). The equivalence classes in non-empty squares are:

$\{$Derby$\}$

$\{$Nottingham$\}$

$\{$Solihull, Sutton Coldfield, Meriden$\}$

$\{$Rugby, Frankton, Coventry$\}$

$\{$Northampton$\}$

$\{$Peterborough, St. Neots$\}$

$\{$Bedford$\}$

$\{$Swindon$\}$

$\{$Oxford$\}$

$\{$Reading, Berkshire$\}$

$\{$Slough$\}$

$\{$London, Borehamwood$\}$

$\{$Salisbury$\}$

$\{$Esher, Epsom, Crawley$\}$.

Finally, for $\alpha=0.95$, each city is in a separate equivalence class.

Tables \ref{similaritytable} and \ref{similaritytable2} show the results of applying the formula for defining proximity relation to the set of the cities.

\begin{landscape}
\begin{table}
	
	\begin{center}
		\begin{tabular}{|c|c|c|c|c|c|c|c|c|c|c|c|}
			\hline
			$s$&1&2&3&4&5&6&7&8&9&10&11\\
			\hline
			1&1.000&0.457&0.540&0.621&0.355&0.741&0.634&0.444&0.238&0.828&0.592 \\
			2&0.457&1.000&0.703&0.332&0.701&0.685&0.818&0.942&0.525&0.524&0.855 \\
			3&0.540&0.703&1.000&0.588&0.807&0.787&0.781&0.648&0.693&0.687&0.792 \\
			4&0.621&0.332&0.588&1.000&0.415&0.621&0.493&0.290&0.387&0.748&0.469\\
			5&0.355&0.701&0.807&0.415&1.000&0.611&0.664&0.645&0.823&0.495&0.697 \\
			6&0.741&0.685&0.787&0.621&0.611&1.000&0.864&0.654&0.479&0.839&0.829 \\
			7&0.634&0.818&0.781&0.493&0.664&0.864&1.000&0.790&0.503&0.704&0.957 \\
			8&0.444&0.942&0.648&0.290&0.645&0.654&0.790&1.000&0.469&0.494&0.820 \\
			9&0.238&0.525&0.693&0.387&0.823&0.479&0.503&0.469&1.000&0.396&0.531 \\
			10&0.828&0.524&0.687&0.748&0.495&0.839&0.704&0.494&0.396&1.000&0.669 \\
			11&0.592&0.855&0.792&0.469&0.697&0.829&0.957&0.820&0.531&0.669&1.000\\
			12&0.551&0.905&0.752&0.415&0.689&0.777&0.913&0.875&0.516&0.617&0.946\\
			13&0.605&0.474&0.751&0.831&0.583&0.720&0.614&0.426&0.548&0.774&0.560\\
			14&0.572&0.379&0.657&0.899&0.498&0.638&0.524&0.332&0.484&0.727&0.507\\
			15&0.562&0.344&0.620&0.919&0.462&0.609&0.492&0.298&0.452&0.709&0.474 \\
			16&0.676&0.404&0.650&0.924&0.469&0.696&0.568&0.364&0.423&0.818&0.543\\
			17&0.482&0.273&0.559&0.856&0.417&0.527&0.413&0.224&0.432&0.623&0.398\\
			18&0.878&0.461&0.610&0.740&0.418&0.775&0.643&0.436&0.321&0.923&0.605\\
			19&0.504&0.953&0.720&0.367&0.689&0.727&0.863&0.922&0.513&0.567&0.897\\
			20&0.652&0.709&0.573&0.371&0.464&0.732&0.790&0.728&0.296&0.620&0.765\\
			21&0.550&0.790&0.579&0.308&0.510&0.685&0.791&0.825&0.335&0.546&0.787\\
			22&0.528&0.559&0.856&0.695&0.716&0.721&0.660&0.505&0.681&0.697&0.660 \\
			23&0.488&0.605&0.896&0.622&0.789&0.709&0.679&0.549&0.743&0.651&0.688\\

			\hline
		\end{tabular}
	\end{center}
\caption{Similarity table}
\label{similaritytable}
\end{table}
\end{landscape}
\begin{landscape}
	
\begin{table}

	\begin{center}
		\begin{tabular}{|c|c|c|c|c|c|c|c|c|c|c|c|c|}
			\hline
			$s$&12&13&14&15&16&17&18&19&20&21&22&23\\
			\hline
			1&0.551&0.605&0.572&0.562&0.676&0.482&0.878&0.504&0.652&0.550&0.528&0.488\\
			2&0.905&0.474&0.379&0.344&0.404&0.273&0.461&0.953&0.709&0.790&0.559&0.605\\
			3&0.752&0.751&0.657&0.620&0.650&0.559&0.610&0.720&0.573&0.579&0.856&0.896\\
			4&0.415&0.831&0.899&0.919&0.924&0.856&0.740&0.367&0.371&0.308&0.695&0.622 \\
			5&0.689&0.583&0.498&0.462&0.469&0.417&0.418&0.689&0.464&0.510&0.716&0.789\\
			6&0.777&0.720&0.638&0.609&0.696&0.527&0.775&0.727&0.732&0.685&0.721&0.709\\
			7&0.913&0.614&0.524&0.492&0.568&0.413&0.643&0.863&0.790&0.791&0.660&0.679\\
			8&0.875&0.426&0.332&0.298&0.364&0.224&0.436&0.922&0.728&0.825&0.505&0.549\\
			9&0.516&0.548&0.484&0.452&0.423&0.432&0.321&0.513&0.296&0.335&0.681&0.743\\
			10&0.617&0.774&0.727&0.709&0.818&0.623&0.923&0.567&0.620&0.546&0.697&0.651\\
			11&0.946&0.560&0.507&0.474&0.543&0.398&0.605&0.897&0.765&0.787&0.660&0.688 \\
			12&1.000&0.547&0.454&0.420&0.489&0.345&0.556&0.950&0.770&0.817&0.614&0.649 \\
			13&0.547&1.000&0.905&0.869&0.875&0.798&0.720&0.504&0.452&0.412&0.863&0.791 \\
			14&0.454&0.905&1.000&0.963&0.893&0.888&0.694&0.410&0.373&0.325&0.781&0.708 \\
			15&0.420&0.869&0.963&1.000&0.886&0.915&0.684&0.375&0.347&0.295&0.745&0.672\\
			16&0.489&0.875&0.893&0.886&1.000&0.804&0.798&0.441&0.445&0.383&0.740&0.670 \\
			17&0.345&0.798&0.888&0.915&0.804&1.000&0.603&0.302&0.263&0.214&0.695&0.625 \\
			18&0.556&0.720&0.694&0.684&0.798&0.603&1.000&0.506&0.593&0.507&0.627&0.577 \\
			19&0.950&0.504&0.410&0.375&0.441&0.302&0.506&1.000&0.748&0.818&0.579&0.620 \\
			20&0.770&0.452&0.373&0.347&0.445&0.263&0.593&0.748&1.000&0.888&0.466&0.474 \\
			21&0.817&0.412&0.325&0.295&0.383&0.214&0.507&0.818&0.888&1.000&0.451&0.474 \\
			22&0.614&0.863&0.781&0.745&0.740&0.695&0.627&0.579&0.466&0.451&1.000&0.927\\
			23&0.649&0.791&0.708&0.672&0.670&0.625&0.577&0.620&0.474&0.474&0.927&1.000\\
			\hline
		\end{tabular}
	\end{center}
\caption{Similarit table (cont.)}
\label{similaritytable2}
\end{table}
\end{landscape}

Using these values, the Shenoi and Melton approach for $\alpha=0.4$, $\alpha=0.6$ and $\alpha=0.8$ gives one equivalence class containing all cities. The reason for this is the density of the data over both coordinate axes. For any chosen pair of the cities there is a sufficient number of other cities between them to enable the merging into one class, even for large values of $\alpha$.

Only for $\alpha$ values very close to 1 the cities start to group into multiple classes. For example, for $\alpha=0.95$ the equivalence classes containing two elements are $\{$Rugby, Frankton$\}$ and $\{$Esher, Epsom$\}$. The other classes are one-member sets.

\section{Conclusion}

The incompleteness or imprecision of the data, as well as the imprecision of the database queries, is best modelled by the fuzzy relational databases. The fuzzy relational model relies, on each scalar domain, on a similarity relation, which is reflexive, symmetric and transitive and induces an equivalence relation on each domain, thus enabling the good properties of the crisp relational databases.

The problem with the fuzzy relational model defined with a similarity relation is the fact that the max-min transitivity property, that is included in the definition of the similarity relation, does not correspond well to some domains, especially the ones with linear ordering.

The fuzzy relational model defined by a proximity relation relaxes the conditions of similarity relation and enables a more natural way of expressing the relationships between the data of linearly ordered domains. At the same time, the transitive closure of the proximity relation enables the establishing of equivalence classes. The equivalence classes are defined on temporal domains and they depend on the content of the database and one possible problem are the situations where the data is ''densely'' distributed, so that the same equivalence class can contain very distant data. The relation obtained by the transitive closure of a proximity relation makes sense when the proximity of elements is defined on some temporally dependent feature, but often does not give good results for domains where the similarity of elements is observed through their defining (and thus unchangeable) characteristics.

The approach to defining equivalence classes by a proximity relation shown in this paper for one-dimensional and two-dimensional domains is independent of database content, based on one-dimensional and two-dimensional intervals, and enables all elements from the same equivalence class to be more proximate than a given value. In case of domains for defining entity characteristics the results are more natural and closer to expectations. The price to be paid is the fact that the proximity of the elements is only a necessary condition for belonging to a same equivalence class. Further research will be focused on reducing this one-sidedness.


\begin{thebibliography}{1}
	
	\bibitem{buckles1} Billy P. Buckles and Frederick E. Petry, A fuzzy representation of data for relational databases, \textit{Fuzzy Sets and Systems} \textbf{7} (1982), 213--226
	
	\bibitem{codd1} E. F. Codd, A relational model of data for large shared data banks,\textit{CACM} \textbf{13} (1970), 377--387.
	
	\bibitem{date1} C. J. Date, \textit{An Introduction to Database Systems}, Sixth Edition, Addison-Wesley, 1994.
	
	\bibitem{koshy1} Thomas Koshy, \textit{Elementary Number Theory with Applications}, Second Edition, Academic Press, 2007.
	
	\bibitem{maier1} David Maier, \textit{The Theory of Relational Databases}, Computer Science Press, Inc., 1983.
	
	\bibitem{petry1} Frederick E. Petry, {\em Fuzzy Databases: Principles and Applications}, Boston: Kluwer Academic Publishers, 1996.
	
	\bibitem{shenoi1} Sujeet Shenoi and Austin Melton, Proximity relations in the fuzzy relational database model, \textit{Fuzzy Sets and Systems} \textbf{31} (1989), 285--296
	
	\bibitem{shenoi2} Sujeet Shenoi, Austin Melton and L. T. Fan, An equivalence classes model of fuzzy
	relational databases, \textit{Fuzzy Sets and Systems} \textbf{38} (1990), 153--170
	
	\bibitem{zadeh1} Lotfi A. Zadeh, Similarity relations and fuzzy orderings, \textit{Information Sciences} \textbf{3} (1971), 177--200
	
	\bibitem{zaniolo1} Carlo Zaniolo, Database relations with null values, \textit{Journal of Computer and System Sciences} \textbf{28} (1984), 142--166
	
	
	
	


	

\end{thebibliography}
\end{document}